\documentclass[reqno]{amsart}%
\usepackage{amsfonts}
\usepackage{amsmath}
\usepackage{amssymb}
\usepackage{graphicx}%
\usepackage{setspace}
\theoremstyle{plain}
\newtheorem{theorem}{Theorem}[section]

\newtheorem{corollary}{Corollary}

\newtheorem{definition}{Definition}

\newtheorem{lemma}[theorem]{Lemma}

\newtheorem{proposition}{Proposition}
\newtheorem{remark}{Remark}

\numberwithin{equation}{section}
\begin{document}
\title[ ]{ Large matrices,Planar graphs and Virasoro Conjecture}
\author[]{Da Xu,   Palle Jorgensen}
\address{Department of Mathematics\\
The University of Iowa\\
Iowa City, IA 52242} \email{dxu@math.uiowa.edu,jorgen@math.uiowa.edu
}
\maketitle
\footnote{The second named author was supported by a grant from the National Science
Foundation.}

\begin{abstract}
 In this paper, we first review one of  difficult parts of the proof of Witten's conjecture by Kontsevich that had not been emphasized before. In the derivation of the KdV equations, we review the boson-fermion correspondence method \cite{K} to show that the trajectory of $\rm GL_\infty$ action on 1 as an element of the ring $\mathbb{C}[x_1,x_2,\cdots]$ yields the solutions of KP hierarchies. Then we  consider the corresponding theory in which the target manifold is a K\"{a}hler manifold. We conjecture that this  nonlinear sigma model is equivalent to a "planar graph" theory. Assuming the conjecture holds, we are able to get the Virasoro constraints in the Virasoro conjecture.  
\end{abstract}

 \maketitle
\section{introduction}
In this paper we consider a problem from string theory whose solution involves moduli spaces from algebraic geometry, unitary representations of infinite-dimensional Lie groups arising as central extensions.
  We shall adopt standard terminology from algebraic geometry, for example, "a moduli space of a Riemann surface of genus $g$ with $n$ points punched" is a geometric space which is the collection of  the complex structure with the $n$ points on the Riemann surface \cite{Mirror}. Such spaces arise generally (as in our present analysis) as solutions to classification problems: For example, if one can show that a collection of smooth algebraic curves of a fixed genus can be given the structure of a geometric space, this then leads to a  new parametrization; so an object viewed  as an entirely separate space. This in turn is accomplished by introducing coordinates on the resulting space. In this context, the term "modulus" is used synonymously with "parameter"; moduli spaces are understood as spaces of parameters rather than as spaces of objects.

Examples: The real projective space $\mathbb{R}P^n$ is a moduli space. It is the space of lines in $\mathbb{R}^{n+1}$ which pass through the origin. Similarly, complex projective space is the space of all complex lines in $\mathbb{C}^{n+1}$. More generally, the Grassmannian $Gk(V)$ of a vector space $V$ over a field $F$ is the moduli space of all k-dimensional linear subspaces of V. The Hilbert scheme $\operatorname{Hilb}(X)$ is a moduli scheme. Every closed point of $\operatorname{Hilb}(X)$ corresponds to a closed subscheme of a fixed scheme $X$, and every closed subscheme is represented by such a point.

Moduli spaces are defined more generally in terms of the moduli functors, and spaces representing them, as is the case for the classical approaches and problems using Teichm\"{u}ller spaces in complex analytical geometry. Our presentation here will take place within this generally framework of moduli spaces.

On the physics side, due to intensive study of string theory, in particular, nonlinear sigma model, the intersection number of moduli spaces play a fundamental role, see \cite{Mirror}, where the readers can find other references.
\subsection{Notation}
Let us recall some definitions in this subsection.\\
$\bar{\mathcal{M}}_{g,n}$ is the compactification of the parameter space of the complex structures  of a Riemann surface with genus $g$ and $n$ points punched out. \\
$\mathcal{L}_i$ denotes the line bundle on $\bar{\mathcal{M}}_{g,n}$ ,which fiber at point $\{C;x_1,x_2,\cdots,x_n \}$ is the cotangent space $T^*_{x_i}(C)$.\\
 $c_1(\mathcal{L}_i)$ denotes the first chern class of the line bundle $\mathcal{L}_i$. \\
When $n=0$, $\bar{\mathcal{M}}_{g,0}$ is the quotient space of the space of metrics $\mathcal{G}_g$ by the group action of $(\operatorname{Diff}\times \operatorname{Weyl}$, where $\operatorname{Diff}$ and $\operatorname{Weyl}$ denotes diffeomorphism and Weyl transformations \cite{P}.\\
The compactification of $\mathcal{M}_{g,n}$ is denoted by $\bar{\mathcal{M}}_{g,n}$ \cite{D}. \\
The following definitions are often used in this paper:
\\

\emph{Amplitude}:
$$\langle \tau_{d_1} \tau_{d_2}
\cdots \tau_{d_n}\rangle=\int_{\bar{\mathcal{M}}_{g,n}} \Pi_{i=1}^n (c_i(\mathcal{L}_i))^{d_i}.
$$
Another notation that will be used is
$$\langle \tau_0^{r_0}\tau_1^{r_1}\tau_2^{r_2}\cdots \rangle:=\langle \tau_{d_1} \tau_{d_2}
\cdots \tau_{d_n}\rangle,
$$
 where $r_0$ of the $d_i$s are equal to 0;$r_1$ of them are equal to 1,etc.\\
\emph{Free energy of genus $g$}:\begin{align}
F_g(t_0,t_1,\cdots)=\langle \exp(\sum_{i=0}^\infty t_i \tau_i)\rangle=\sum_{(k)}\langle \tau^{k_0}_0 \tau^{k_1}_1
\cdots \rangle \prod_{i=0}^\infty \frac{t^{k_i}_i}{k_i !}, \label{1}
\end{align}\\
\emph{Partition function}:
\begin{align}
Z(t_*)=\exp \sum_{g=0}^\infty F_g(t_*),
\end{align}
 where  the free energy $F_g(t_*)$.

\subsection{Witten's 1990 conjecture }

Witten's  conjecture (1990) asserts that \emph{the partition function $Z(t_*)=\exp \sum_{g=0}^\infty F_g(t_*)$ \cite{W1} is the $\tau$-function of the KdV hierarchy}. A $\tau$-function for the KdV hierarchy means from the $\tau$-function, we can construct the solution of the KdV equation:
\begin{align}
\frac{\partial U}{\partial t_1}=U\frac{\partial U}{\partial t_0}+\frac{1}{12}\frac{\partial^3 U}{\partial t_0^3},\label{KdV}
\end{align}
where $U=\frac{\partial^2 \ln Z(t_*)}{\partial t_0^2}$.

The Korteweg-de Vires equations(KdV) have their origins in the story of water waves in a shallow channel, and have numerous of applications different from their origin.
If $x$ is the space variable, the standard form of KdV is
\begin{align}
u_t+u_{xxx}+6uu_x=0.  \label{KdV1}
\end{align}
This equation is known to have soliton solutions, i.e., $u(x,t)=f(x-ct)$ where
$f$ satisfies
\begin{align}
-cf^{'}+f^{'''}+6ff^{'}=0. \label{KdV2}
\end{align}
so,
\begin{align}
f=\frac{c}{2\cosh^2(\frac{\sqrt{c}}{2}(x-a))}\nonumber,
\end{align}
where $a$ and $c$ are constants. Since this initial discovery, other  after invariants have been found. The first such sequence of invariants
came in 1960s from P.Lax's commutator method: Let $u=u(x,t)$ be a function in two variables, and consider
\begin{align}
(L_u(t)f)(x)=-f^{''}+u(x,t)f(x),\nonumber
\end{align}
i.e.,
the Schr\"{o}dinger operator
\begin{align}
L_u(t)=-(\frac{d}{dx})^2+u(x,t) \label{KdV3}
\end{align}
acting on function $f(x)$.

Lax (1968) found that if
\begin{align}
A=4(\frac{\partial}{\partial x})^3-3 (u\frac{\partial}{\partial x}+\frac{\partial}{\partial x}u). \nonumber
\end{align}
Then
\begin{align}
\frac{\partial}{\partial t}L_u(t)=[L_u,A]=L_u A-AL_u,\label{KdV4}
\end{align}
and this accounts for one infinite family of "integrals" or invariants. Specifically with $u=u(x,t)$, consider
$L=L_u$ as in (\ref{KdV3}), we then obtain that eigenvalues of (\ref{KdV3}) yield invariants, and the relevant $u$ in (\ref{KdV4}) is from the solutions to the KdV equation (\ref{KdV1}).
\subsection{KdV hierarchy}

Kontsevich's proof mainly consists of three steps:\\
1.Based on a theorem of Strebel, the one to one correspondence of the space $\bar{\mathcal{M}}_i \times \mathbb{R}_{+}^n$  with  the "fat graphs". And then the main identity is proved:
\begin{align}
\sum_{d_*:\sum d_i=d}\langle \tau_{d_1}\cdots \rangle \prod_{i=1}^n \frac{{(2d_i-1)!!}}{\lambda_i^{2d_i+1}}
=\sum_{\Gamma\in G_{g,n}} \frac{2^{-v(\Gamma)}}{|Aut \Gamma|} \prod_{e\in E}\frac{2}{\tilde{\lambda}(e)}. \label{main}
\end{align}\\
2.By the main identity and the Feynman diagram techniques, the partition function which is the exponential of the free energy $F(t_0(\Lambda),t_1(\Lambda),\cdots)$, where
$$
t_i(\Lambda)=-(2i-1)!! tr \ (\Lambda^{-(2i+1)})
$$
is the asymptotic expansion of the random matrix integral
\begin{align}
I_N(\Lambda)=\int \exp(\frac{\sqrt{-1}}{6} tr \ (M^3) d\mu_{(\Lambda)}(M),\label{2}
\end{align}
where the measure $d\mu_{(\Lambda)}(M)=\frac{- \frac{tr \ (M^2 \Lambda)}{2}dM }{\int - \frac{tr \ (M^2 \Lambda)}{2}dM }$.\\

3. By expansion of matrix Airy function and some properties of $\tau$-functions of KdV hierarchy,the author shows the
integral (\ref{2}) is the $\tau$-function is the asymptotic expansion of (\ref{2}) as the rank of $\Lambda$ goes to infinity.
\subsection{the first step}

A \emph{quadratic differential} $\phi$ on a Riemann surface $C$ of finite type is a holomorphic section of the line bundle $(T^*)^{\otimes 2}$. A nonzero quadratic differential defines a metric in a local coordinates $z$:
\begin{align}
|\phi(z)|^2 |dz|^2, \ where \ \phi=\phi(z)dz^2.
\end{align}
A \emph{horizontal trajectory}  of a quadratic differential is a curve along which $\phi(z)dz^2$ is real and positive. JS quadratic differentials are those for which the union of nonclosed trajectories has measure zero.
Strebel proved in 1960s the following
\begin{theorem}
For any connected Riemann surface $C$ and $n$ distinct points $x_1,\cdots,x_n\in C, n>0,n>\chi(C)$ and $n$ positive real numbers $p_1,p_2,\cdots,p_n$  there exists a unique $JS$ quadratic differential on $C/ \{x_1,\cdots, x_n \}$ whose maximal ring domains are $n$ punctured disks $D_i$  surrounding  points $x_i$ with circumference $p_i$.
\end{theorem}
 Based on Strebel's theorem, Kontsevich found the one to one correspondence between the fat graphs, which are formed by the closed horizontal trajectories, which end at the zeros of the JS forms.  and the product space $\bar{\mathcal{M}}_{g,n}\times \mathbb{R}^n_{+}$. Each graph carries the following structures\\
(1)for each vertex a cyclic order on the set of germs of edges meeting this vertex is fixed;\\
(2) to each edge is attached a positive real number, its length(which is determined by the metric);\\
(3)the valency of each vertex of a fat graph is  three(we can derive that each valency is at least 3 by changing to polar system);\\
(4) the loops of the graph is numbered by $1,2,\cdots,n$;\\
(5) we make these graphs double-line graphs(this is not required by the one to one correspondence theorem).\\
For a fat graph, denote $l_e$ the length of a edge(double) $e$ and for each face $f$, the perimeter $p_f=\sum_{e \subset f}l_e$. Then we have
\begin{align}
E-n-V=2g-2,
\end{align}
\begin{align}
2E=3V.
\end{align}
Kontsevich proved the first chern class $c_1(\mathcal{L}_i)$ can be written as (this step is not hard)
\begin{align}
\omega_i=\sum_{a,b\in f_i} d(l_1/p_i)\wedge d(l_b/p_i),
\end{align}
and we can define a volume form on the fat graph space
which is $\Omega^d/d!$, where $\Omega=\sum_i p_i^2 \omega_i$.

Now since the volume form on the fat graph space $M^{comb}$ is defined, we can compute  the Laplace transform with respect to $p_1,p_2,\cdots,p_n$
\begin{align}
& \prod_{i=1}^n (\int_0^\infty e^{-\lambda_i p_i}dp_i) \int \frac{(\sum_i p_i^2 \omega_i)^{d}}{(d)!}\nonumber\\
=& 2^{d} \sum_{d_1+d_2+\cdots+d_n=d} \langle \tau_{d_1}\cdots\tau_{d_n}\rangle \prod_{i=1}^n (2d_i-1)!! \lambda_i^{-2d_i-1}.
\end{align}
On the other hand, by a very very delicate argument on complex cohomology, one has
\begin{align}\label{1.9}
1/d! \prod_i dp_i \wedge (\sum_i \sum_{a,b\subset f_i} dl_1\wedge dl_b)^d=2^{5g-5+2n} dl_1\wedge dl_2\cdots dl_E.
\end{align}
In the right hand side, we endow a orientation. Therefore we get the main identity
\begin{align}
\sum_{d_*:\sum d_i=d}\langle \tau_{d_1}\cdots \rangle \prod_{i=1}^n \frac{{(2d_i-1)!!}}{\lambda_i^{2d_i+1}}
=\sum_{\Gamma\in G_{g,n}} \frac{2^{-V(\Gamma)}}{|Aut \Gamma|} \prod_{e\in E}\frac{2}{\tilde{\lambda}(e)}. \label{main}
\end{align}

In the derivation of the factor in the right hand side of (\ref{1.9}), the torsion of chain complex plays an important role(see Appendix C \cite{K}).  The following  definitions and theorems are useful.
\begin{definition}
The chain complex $C$ is said to be acyclic if the homology $H_i(C)=0$ for all $i$. The chain complex $C$ is said to be based if each $C_i$ has a distinguished basis $c_i$.
\end{definition}
Then the torsion is defined by
\begin{definition}
The torsion of $C$ is
\begin{align}
\tau(C)= \prod_i [b_i b_{i-1}/c_i]^{(-1)^{i+1}} \in \mathbb{F}^*.
\end{align}
\end{definition}

Let $0\rightarrow C'\rightarrow C \rightarrow C^{''}\rightarrow 0$ be a short exact sequence of chain complexes. For a fixed $i$, we have a short exact sequence $A_i$
\begin{align}
0\rightarrow C_i'\rightarrow C_i\rightarrow C_i^{''}.
\end{align}
The following formula is straightforward
\begin{align}
\tau(C)=(\text{sign} \prod_i \tau(A_i)) \tau(C')\tau(C^{''}).
\end{align}

 Let $V$ be a vector space over the field $\mathbb{F}$. Then let $\Omega_i=\wedge^i V$. For an acyclic chain complex
$C$, if each $C_i$ is a linear subspace of $\Omega_i$, we can define a generalized torsion on this chain complex
\begin{definition}
\begin{align}
\tau(C)=[b_i \wedge b_{i-1}/(\wedge_{i=1}^n e_i ]^{(-1)^{i+1}},
\end{align}
where $b_i$ is the wedge product of the basis elements of $B_i=\text{Im} (\partial_i: C_{i+1}\rightarrow C_i)$. $\{e_i\}$ is a basis of the vector space $V$.
\end{definition}

By the virtue of the following theorem, we are able to transfer the torsion of $C\otimes \mathbb{R}$ to a computation of cohomology of $C$ \cite{Tu}:
\begin{theorem}
Let $R$ be a Noetherian unique factorization domain. Let $C=(C_m\rightarrow\cdots\rightarrow C_0)$ be a based free chain complex of finite rank over $R$ such that $\text{rk} H_i(C)=0$ for all $i$. Let $\tilde{R}$ be the field of fractions of $R$. Then the based chain complex $\tilde{C}=\tilde{R}\otimes_R C$ is acyclic and
\begin{align}
\tau(\tilde{C})=\prod_{i=0}^m (\operatorname{ord} H_i(C))^{(-1)^{i+1}}.
\end{align}
\end{theorem}

\section{KP and KdV hierarchy}

\subsection{Solutions of KP hierarchy as the orbit of $GL_{\infty}$ action on $\mathbb{C}[x_1,x_2,\cdots,]$}
In Kac's presentation, the infinite dimensional group $\rm GL_{\infty}$ has a representation on infinite wedge space $\Lambda^\infty V$. By the fermion-boson correspondence,
$\rm GL_{\infty}$ has a representation on the space $B=\mathbb{C}[x_1,x_2,\cdots]$, the polynomial ring of infinite many variables.  Then we have a representation of
$\rm GL_\infty$ on $B$. Then the KP hierarchy is the orbit $\Omega$ of the the vacuum 1 in $B$ under the action of $\rm GL_\infty$, i.e., $\Omega=\rm GL_\infty \cdot 1$. Also, Dirac's positron theory can be given a representation-theoretic interpretation and used to obtain highest weight representations of these Lie algebras. The following are the relevant definitions and theorems.

\textbf{Representation theory}.We will follow Jorgensen's book \cite{J} and Kac's book \cite{K} to derive KP hierarchy. We offer a simple philosophy to get the the KP hierarchy  , which is a set of infinite many PDEs.
Let $\mathcal{A}$ be the Heisenberg algebra, the complex Lie algebra with a basis $\{a_n,n\in\mathbb{Z};\hbar\}$, with the commutation relations
\begin{align}
[\hbar, a_n]& =0,  \ (n\in\mathbb{Z}),\nonumber\\
[a_m,a_n]&= m \delta_{m,-n}\hbar  \   (m,n\in\mathbb{Z}).
\end{align},
and $\{L_n\}$ denotes the Virasoro algebra with central extension $c$,i.e.,
\begin{align}
[L_n,L_m]=(n-m)L_{n+m}+\frac{\delta_{m,-n}}{12}c.
\end{align}

One thing that is nice in the context of our two infinite systems of
operators $a_n$ and $L_k$ below in sections 4 and 5 is the following close
analogue to an important family of unitary representations of Lie groups. It
is in fact a natural extension of what was first realized for finite
dimensional groups as the Weil-Segal-Shale representations of the
metaplectic groups. Details below: Consider the following setting for a
finite-dimensional Heisenberg group $H$. Let $G$ be the corresponding group of
automorphisms of $H$ which fix the center.
Then $G$ is a finite dimensional Lie group, the metaplectic group. Pick a
Schr\"{o}edinger representation of $H$, and compose it with an automorphism, so an
element in $G$. The result is a second representation of $H$. By the Stone-von
Neumann uniqueness theorem the two representations are unitarily equivalent,
and so the equivalence is implemented by a unitary operator $U(g)$. By passing
to a double cover of $\tilde{G}$ one can show that $U(g)$ in fact then defines a
unitary representation of $\tilde{G}$.
   If we now pass to the corresponding Lie algebras $L(H)$ and $L(G)$ we see
that $L(H)$ is normalized by $L(G)$. Moreover $L(H)$ in the Schr\"{o}edinger
representation is spanned by Heisenberg's canonical operators $P$, $Q$, and the
one-dimensional center; here we write $P$ for momentum and $Q$ for position,
possibly with several degrees of freedom.

\textbf{Highest Weights.} By comparison, in the Weil representation, the Lie algebra $L(G)$ is then spanned
by all the quadratic polynomials in the $P$s and the $Q$s.
  Now the Stone-von Neumann uniqueness theorem is not valid for an infinite
number of degrees of freedom, but nonetheless, the representations of the
two Lie algebras we present by the infinite systems of operators $a_n$ and $L_k$
in sections 4 present themselves as a close analogy to the Weil
representations in the case of Lie groups, i.e., the case of a finite number
of degrees of freedom. Our infinite-dimensional Virasoro Lie algebra spanned
by the infinite system $\{a_n\}$, and it is a central extension; hence a direct
analogue of the Heisenberg Lie algebra. Similarly, our infinite-dimensional
Lie algebra of quantum fields spanned by $\{L_k\}$ normalizes the Virasoro Lie
algebra, and so it is a direct analogue of the Lie algebra of operators $L(G)$
in the finite-dimensional case.
    In both of these cases of representations, the operators in the
respective Lie algebras are unbounded but densely defined in the respective
infinite-dimensional Hilbert spaces. We show that our representations of the
Lie algebra of quantum fields spanned by $\{L_k\}$ may be obtained with the use
of highest weight vectors, and weights.

Define the Fock space $B:=\mathbb{C}[x_1,x_2,\cdots]$.
Given $\mu,\hbar\in \mathbb{R}$, define the following representation of $\mathcal{A}$ on $B$ ($n\in\mathbb{N}$):
\begin{align}
a_n &=\epsilon_n \partial/\partial x_n,\nonumber\\
a_{-n} &=\hbar \epsilon^{-1}_n n x_n,\nonumber\\
a_0 &=\mu I,\nonumber\\
\hbar &=\hbar I.
\end{align}
When $\hbar\neq 0$, the representation is irreducible, since one can get any polynomial

Let $V=\oplus_{j\in\mathbb{Z}} \mathbb{C}v_j$ be an infinite dimensional vector space over $\mathbb{C}$ with a basis $\{v_j;j\in \mathbb{Z}\}$.

The Lie algebra $gl_\infty :=\{(a_{ij})_{i,j\in\mathbb{Z}}; \text{all but a finite number of the $a_ij$ are 0}\}$,
with the Lie bracket being the ordinary matrix commutator.

The Lie algebra $gl_\infty$ is the Lie algebra of the Lie group
\begin{align}
\rm GL_\infty=\{A=(a_{ij})_{i,j\in\mathbb{Z}};\text{$A$ invertible and all but a finite number of $a_{ij}-\delta_{ij}$ are zero} \}
\end{align}.
The group action is matrix multiplication.
Define the shift operator $\Lambda_k$ by
\begin{align}
\Lambda_k v_j=v_{j-k}.
\end{align}

Then the representation of $Vect$ in $V_{\alpha,\beta}$ in the vector space $V$ can be
\begin{align}
L_n(v_k)=(k-\alpha-\beta(n+1))v_{k-n},
\end{align}
which implies
\begin{align}
L_n=\sum_{k\in\mathbb{Z}} (k-\alpha-\beta(n+1))E_{k-n,k}.
\end{align}
Then $L_n\in \bar{a_\infty}$.

\begin{definition}
The elementary Schur polynomials $S_k(x)$ are polynomials belonging to $\mathbb{C}[x_1,x_2,\cdots]$ and are defined by the generating function
\begin{align}
\sum_{k\in\mathbb{Z}} S_k(x) z^k= \exp(\sum_{k=1}^\infty) x_k z^k.
\end{align}
\end{definition}

We list the following propositions and theorems in Kac's book \cite{K} without proofs. These propositions or theorems are useful in this paper.
Then in \cite{K}, there are the following
\begin{definition}
The generating series are defined by
\begin{align}
X(u):=\sum_{j\in\mathbb{Z}} u^j \hat{v}_j,  \   X^*(u)=\sum_{j\in\mathbb{Z}} u^{-j} \check{v}_j^*,
\end{align}
where $u$ is a nonzero complex number.
 \end{definition}

\begin{proposition}
$\Gamma(u)$ and $\Gamma^*(u)$ have the following form on $B^{(m)}$:
\begin{align}
\Gamma(u)|_{\hat{B}^{(m)}}=u^{m+1} z \exp(\sum_{j\geq 1}u^j x_j)\exp(-\sum_{j\geq 1} \frac{u^{-j}}{j}\frac{\partial}{\partial x_j}),\nonumber\\
\Gamma^*(u)|_{\hat{B}^{(m)}}=u^{-m} z^{-1} \exp(-\sum_{j\geq 1}u^j x_j)\exp(\sum_{j\geq 1} \frac{u^{-j}}{j}\frac{\partial}{\partial x_j}).
\end{align}
\end{proposition}

Then the representation $g\ell_\infty$ can be determined by the isomorphism
\begin{align}
\sigma_m:F^{(m)}\rightarrow B^{(m)}.
\end{align}

$E_{ij}$ is represented by $\hat{v}_i\check{v}_j$.

Consider the generating function
\begin{align}
\sum_{i,j\in\mathbb{Z}}u^{i}v^{-j}E_{ij}.
\end{align}

The representation in $\hat{F}$ of this generating function under $r$ is
\begin{align}
\sum_{i,j\in\mathbb{Z}} u^i v^{-j} r^B(E_{ij})\equiv \sigma_m(X(u)X^*(v))\sigma_m^{-1}=\frac{(u/v)^m}{1-(v/u)}\Gamma(u,v),
\end{align}
where $\Gamma(u,v)$ is the vertex operator
\begin{align}
\Gamma(u,v)=\exp(\sum_{j\geq 1}(u^j-v^j)x_j)\exp(-\sum_{j\geq 1}\frac{u^{-j}-v^{-j}}{j}\frac{\partial}{\partial x_j}).
\end{align}

\begin{proposition}
If $\tau\in\Omega$, then $\tau$ is a solution of the equation
\begin{align}
\sum_{j\in\mathbb{Z}}  \check{v_j}(\tau)\otimes \hat{v_j}(\tau)=0. \label{7.2}
\end{align}
Conversely, if $\tau\in F^{(0)}$, $\tau\neq 0$ and $\tau$ satisfies (\ref{7.2}), then $\tau\in\Omega$.
\end{proposition}
\begin{proposition}
The Schur polynomials $S_{\lambda}(x)$ are contained in $\Omega$.
\end{proposition}

\begin{proposition}
A nonzero element $\tau$ of $\mathbb{C}[x_1,x_2,\cdots]$ is contained in $\Omega$ if and only if the coefficient of $\mu^0$ vanishes in the expression:
\begin{align}
u\exp(-\sum_{j\geq 1} 2u^j y_j)\exp(\sum_{j\geq 1}
\frac{u^{-j}}{j}\frac{\partial}{\partial y_j})\tau(x-y)\tau(x+y).
 \end{align}
\end{proposition}

\begin{theorem}(Kashiwara and Miwa,1981)
A nonzero polynomial $\tau$ is contained in $\Omega$ if and only if $\tau$ is a solutioin of the following system of Hirota bilinear equations:
\begin{align}
\sum_{j=0}^\infty S_j(-2y)S_{j+1}(\tilde{x})\exp(\sum_{k\geq 1}y_k x_k)\tau(x)\cdot \tau(x)=0,
\end{align}
where $y_1,y_2,\cdots$ are free parameters.
\end{theorem}

Then it becomes the Kadomtzev-Petviashvili(KP) equation:
\begin{align}
\frac{3}{4}\frac{\partial^2}{\partial y^2}=\frac{\partial}{\partial x}(\frac{\partial u}{\partial t}-\frac{3}{2}u \frac{u}{\partial x}-\frac{1}{4}\frac{\partial^3 u}{\partial x^3}).
\end{align}

\begin{corollary}
The functions $2\frac{\partial^2}{\partial x^2}(\log S_\lambda(x,y,t,c_4,c_5,\cdots))$, where $c_4,c_5,\cdots$ are arbitrary constants, are the solutions of  the $KP$ equation.
\end{corollary}

\subsection{highest weight condition}
The $\tau$-function of the KdV hierarchy  is annihilated by a sequence of
differential operators, which form a half branch of the Virasoro algebra.
(\cite{D},\cite{F}, and \cite{Kac1}).  For the partition function here
According to \cite{K}, the Virasoro algebra with central charge $c_\beta$ can be represented by
\begin{align}
L_i & =\hat{r}(d_i) \ \text{if} \ i\neq 0,\nonumber\\
L_0 & =\hat{r}(d_0)+h_0,
\end{align}
where $c_\beta=-12\beta^2+12\beta-2$, $h_m=\frac{1}{2}(\alpha-m)(\alpha+2\beta-1-m)$.

Now we can compute the highest weight condition: under what condition, a function $Z\in \mathbb{C}[x_1,x_2,\cdots,]$ can be annihilated by $L_n$, $n\geq -1$.
In fact, when $i\neq 0$,
\begin{align}
L_i=\hat{r}(d_i)&=\hat{r}(\sum_{k\in\mathbb{Z}}(k-\alpha-\beta(i+1))E_{k-i,k})\nonumber\\
               &=\sum_{k\in\mathbb{Z}}(k-\alpha-\beta(i+1))\hat{r}(E_{k-i,k})\nonumber\\
               &=\sum_{k\in\mathbb{Z}}(k-\alpha-\beta(i+1))r(E_{k-i,k}),
\end{align}
and when $i=0$,
\begin{align}
L_0 &=\hat{r}(d_0)+h_0\nonumber\\
    &=\hat{r}(\sum_{k\in\mathbb{Z}}(k-\alpha-\beta)E_{k,k})+h_0\nonumber\\
    &=\sum_{k\in\mathbb{Z}}(k-\alpha-\beta)\hat{r}(E_{k,k})+h_0\nonumber\\
    &=\sum_{k>0}(k-\alpha-\beta)(r(E_{k,k})-I)+\sum_{k\leq 0}(k-\alpha-\beta)r(E_{k,k})+h_0\nonumber\\
\end{align}

Moreover, since the transported representation $\hat{r}^B_m=\sigma_m\hat{r}_m \sigma_{m}^{-1}$ of $\mathcal{A}$ on $B^{(m)}$ is
\begin{align}
\hat{r}^B_m(\Lambda_k)&=\frac{\partial}{\partial x_k}, \nonumber\\
\hat{r}^B_m (\Lambda_{-k})& =k x_k,\nonumber\\
\hat{r}^B_m (\Lambda_0)&=m,
\end{align}
which is a representation of Heisenberg algebra on $B^{(m)}$. Then by a result of Fairlie \cite{CT}, there is a oscillator representation of Virasoro algebra for arbitrary $\lambda$, $\mu$
\begin{align}
L_0 &=(\mu^2+\lambda^2)/2+\sum_{j>0} a_{-j}a_j,\nonumber\\
L_k &=\frac{1}{2}\sum_{j\in\mathbb{Z}}a_{-j}a_{j+k}+i\lambda ka_k,
\end{align}
for $k\neq 0$(take $\hbar =1$ and $a_0=\mu$).
It is easy to verify that the central charge for this Virasoro algebra is $1+12\lambda^2$.

Then in this case the representation of the Virasoro algebra on $B^{(m)}$  is
\begin{align}
L_0 &=(\mu^2+\lambda^2)/2+\sum_{j>0}jx_j \frac{\partial}{\partial x_j},\nonumber\\
L_k &=\frac{1}{2}\sum_{j\in\mathbb{Z}}jx_j\frac{\partial}{\partial x_{j+k}}+i\lambda k\frac{\partial}{\partial x_k}, \ k\geq 0,\nonumber\\
L_k &=\frac{1}{2}\sum_{j\in\mathbb{Z}}jx_j\frac{\partial}{\partial x_{j+k}}+i\lambda k^2 x_k, \ k\leq 0.
\end{align}

\subsection{the difficulty of a conjecture of Kontsevich}
Kontsevich also proposed some conjectures in \cite{K}. Let us see the some  of them that are concerned with the KdV hierarchies.
First of all, one can introduce variables $s$ :
\begin{align}
Z(t_0,t_1,\cdots,;s_0,s_1,\cdots)=\exp(\sum_{n_* m_*} \langle \tau_{d_1}\cdots\tau_{d_n}\rangle_{m_0,m_1,\cdots} \prod_{i=0}^{\infty} \frac{t_i^{n_i}}{n_i!}\prod_{j=0}^\infty s_j^{m_j}).
\end{align}
It can be shown (\cite{K}) that  that $Z(t_*(¦«), s_*)$  is an asymptotic
expansion of
\begin{align}
I_N(\Lambda)=\int \exp(\sqrt{-1}\sum_{j=0}^{\infty} (-1/2)^j s_j \frac{\text{tr} \ M^{2j+1} }{2j+1})d\mu_\Lambda(M).
\end{align}
Then we can list the statements of these conjectures are\\
\emph{
  1. $Z(t,s)$ is a $\tau$-function for KdV-hierarchy in variables $T_{2i+1}:=\frac{t_i}{(2i+1)!!}$ for arbitrary $s$.\\
  2. $Z(t,s)$ is a $\tau$-function for KdV-hierarchy in variables $T_{2i+1}:=\frac{s_i}{(2i+1)!!}$ for arbitrary $t$.\\
  3. Let $T$ be any formal $\tau$-function for the KdV-hierarchy considered as a matrix function. Then $\int T(X)d\mu_{\Lambda}(X)$ is a matrix $\tau$-function for the KdV-hierarchy in $\Lambda$.}

We shall explain the difficulty of the first conjecture in this subsection.
We recall the Harish-Chandra formula \cite{H}.
\begin{lemma}\label{4.1}
 If $\Phi$ is a conjugacy invariant function on the space of hermitian $N\times N$-matrices, then for any diagonal hermitian matrix $Y$,
\begin{align}
\int \Phi(X) e^{-\sqrt{-1} \text{tr} XY}dX=(-2\pi \sqrt{-1})^{N(N-1)/2} (V(Y))^{-1}\int \Phi(D) e^{-\sqrt{-1}\text{tr} DY}V(D)dD,
\end{align}
where the last integral is taken over the space of diagonal hermitian matrices $D$;$V$ is the Vandermonde Polynomial determinant which is defined by
\begin{align}
V(\text{diag}(X_1,X_2,\cdots,X_n):= \prod_{i<j} (X_j-X_i)=\det (X_i^{j-1}).
\end{align}
\end{lemma}
Harish-Chandra generalized the above fact: \cite{Harish}:

Let $G$ be a compact simple Lie group, $L$ its Lie algebra of order $N$ and rank $n$, $W$ the Weyl group of $L$,$R_{+}$ the set of positive roots, and $m_i=d_i-1$ its Coxerter indexes. Also $X$ and $Y$ elements of $L$. Let $(X,Y)$ be a bilinear form which is invariant under $G$,i.e.,$(gX,gY)=(X,Y)$,for $\forall g\in G$. Then
\begin{align}
\int_{g\in G} \exp(c(X,gYg^{-1}) dg=\text{const} \sum_{w\in W} \epsilon_w \exp(c(X,wY)/\prod_{\alpha\in R_{+}} (\alpha,X)(\alpha,Y).
\end{align}
 We have the following
\begin{lemma}\label{4.3}
\begin{align}
\int \Phi(X) e^{- \frac{1}{2}\text{tr} \Lambda X^2}dX=\text{const}\int \Phi(D)V(D) \frac{\sum_{w\in S_N} \text{sign}(w) e^{-\frac{1}{2}\text{tr}(\Lambda  w(D^2))}}{ V(\Lambda)\prod_{i<j}(D_i+D_j)} dD.
\end{align}
\end{lemma}
\begin{proof}
We apply Harish-Chandra's result to the unitary group $\rm U(N)$, then
\begin{align}
& \int \Phi(X) e^{- \frac{1}{2}\text{tr} \Lambda X^2}dX \nonumber\\
=& \text{const}\int \Phi(D) (\int e^{- \frac{1}{2}\text{tr} \Lambda UD^2 U^{-1}}dU)V^2(D) dD\nonumber\\
=& \text{const}\int \Phi(D) \nonumber\\
& \cdot(\sum_{w\in S_N} \text{sign}(w) e^{-\frac{1}{2}\text{tr}(\Lambda U^{-1} D^2 U)}/\prod_{1\leq i<j\leq N}\text{tr}((\epsilon_j-\epsilon_i)\Lambda)\text{tr}((\epsilon_j-\epsilon_i)D^2)V^2 (D)dD\nonumber\\
=& \text{const}\int \Phi(D) \nonumber\\
& \cdot(\sum_{w\in S_N} \text{sign}(w) e^{-\frac{1}{2}\text{tr}(\Lambda U^{-1} D^2 U)}/\prod_{1\leq i<j\leq N}\text{tr}((\epsilon_j-\epsilon_i)\Lambda)\text{tr}((\epsilon_j-\epsilon_i)D^2)V^2(D)dD\nonumber\\
=& \text{const}\int \Phi(D)V(D) \frac{\sum_{w\in S_N} \text{sign}(w) e^{-\frac{1}{2}\text{tr}(\Lambda  w(D^2))}}{ V(\Lambda)\prod_{i<j}(D_i+D_j)} dD\nonumber\\
\end{align}
\end{proof}

We need the following
\begin{lemma}\label{field}
$I_N(\Lambda)$ is symmetric with respect to $\Lambda$ and $I_N(\Lambda)$ is in the field $\mathbb{C}(\Lambda)$, which is the field of polynomial ring of $\Lambda$.
\end{lemma}
\begin{proof}
By direct computation,
\begin{align}
\int \exp(-\frac{1}{2}\text{tr} \Lambda M^2)dM=2^{\frac{N(N-1)}{2}}(2\pi)^{\frac{N^2}{2}} \prod_{r=1}^N \lambda_r^{-\frac{1}{2}} \prod_{i<j}(\lambda_i+\lambda_j)^{-1}. \label{4.8}
\end{align}
By Lemma (\ref{4.3}), (\ref{4.8}) can be written as
\begin{align}
\text{const}\int V(D) \frac{\sum_{w\in S_N} \text{sign}(w) e^{-\frac{1}{2}\text{tr}(\Lambda  w(D^2))}}{ V(\Lambda)\prod_{i<j}(D_i+D_j)} dD=2^{\frac{N(N-1)}{2}}(2\pi)^{\frac{N^2}{2}} \prod_{r=1}^N \lambda_r^{-\frac{1}{2}} \prod_{i<j}(\lambda_i+\lambda_j)^{-1}\nonumber,
\end{align}
or
\begin{align}
\text{const}\int V(D) \frac{\sum_{w\in S_N} \text{sign}(w) e^{-\frac{1}{2}\text{tr}(\Lambda  w(D^2))}}{ \prod_{i<j}(D_i+D_j)} dD=2^{\frac{N(N-1)}{2}}(2\pi)^{\frac{N^2}{2}}V(\Lambda) \prod_{r=1}^N \lambda_r^{-\frac{1}{2}} \prod_{i<j}(\lambda_i+\lambda_j)^{-1}.
\end{align}
Acting on both sides by operator
$\sum_{i=1}^N \frac{\partial}{\partial \lambda_i}$ and multiplying
$$\frac{2^{-\frac{N(N-1)}{2}}(2\pi)^{-\frac{N^2}{2}} \prod_{r=1}^N \lambda_r^{\frac{1}{2}}, \prod_{i<j}(\lambda_i+\lambda_j)}{V(\Lambda)}
$$
we get
\begin{align}
& = \text{const}\frac{2^{-\frac{N(N-1)}{2}}(2\pi)^{-\frac{N^2}{2}} \prod_{r=1}^N \lambda_r^{\frac{1}{2}} \prod_{i<j}(\lambda_i+\lambda_j)}{V(\Lambda)}\nonumber\\
&\cdot\int (\sum_{i=1}^N D_i^2)V(D) \frac{\sum_{w\in S_N} \text{sign}(w) e^{-\frac{1}{2}\text{tr}(\Lambda  w(D^2))}}{\prod_{i<j}(D_i+D_j)} dD \nonumber\\
&= -\sum_{r=1}^N \frac{1}{\lambda_r}-\sum_{r<k} \frac{2}{\lambda_r+\lambda_k}.
\end{align}
\end{proof}
Similarly, we can prove for all symmetric polynomials $P(\Lambda)$,
$\langle P(\Lambda) \rangle\in \mathbb{C}(\Lambda)$, but not in $\mathbb{C}[\Lambda]$. This is the main reason why
the first conjecture is hard to prove, since the $\tau$-function of KP hierarchy is in $\mathbb{C}[\Lambda]$.

\section{Virasoro conjecture and matrix model}
Let us recall the definition of $\bar{\mathcal{M}}_{g,n}(M,\beta)$ and some properties \cite{Mirror}.
\begin{definition}
Let $M$ be a non-singular projective variety. A morphism $f$ from a pointed nodal curve to $X$ is a stable map if every genus 0 contracted component of $\Sigma$ has at least three special points, and every genus 1 contracted component has at least one special point.
\end{definition}

\begin{definition}
A stable map  represents a homology class $\beta\in H_2(M,\mathbb{Z})$ if $f_*(C)=\beta$.
\end{definition}

The moduli space of stable maps from n-pointed genus $g$ nodal curves to $M$ representing the class $\beta$ is denoted $\bar{\mathcal{M}}_{g,n}(M,\beta)$.
The moduli space $\bar{\mathcal{M}}_{g,n}(M,\beta)$ is a Deligne-Mumford stack. It has the following properties:\\
(1)There is an open subset $\bar{\mathcal{M}}_{g,n}(M,\beta)$ corresponding maps from non-singular curves.\\
(2)$\bar{\mathcal{M}}_{g,n}(M,\beta)$ is compact.\\
(3)There are $n$ "evaluation maps" $\operatorname{ev}_i:\bar{\mathcal{M}}_{g,n}(M,\beta)\rightarrow M$ defined by
\begin{align}
\operatorname{ev}_i(\Sigma,p_1,\cdots,p_n,f)=f(p_i), \   1\leq i\leq n.
\end{align}
(4)If $n_1\geq n_2$, there is a "forgetful morphism"
\begin{align}
\bar{M}_{g,n_1}(M,\beta)\rightarrow \bar{M}_{g,n_2}(M,\beta).
\end{align}
so long as the space on the right exists. \\
(5)There is a "universal map" over the moduli space:
\begin{align}
(\tilde{\Sigma},\tilde{p_1},\cdots, \tilde{p_n})\xrightarrow{\bar{f}} M,\nonumber\\
(\tilde{\Sigma},\tilde{p_1},\cdots, \tilde{p_n})\xrightarrow{\pi}  \bar{\mathcal{M}}_{g,n}(M,\beta).
\end{align}

(6)Given  a morphism $g:X\rightarrow Y$, there is an induced morphism
\begin{align}
\bar{\mathcal{M}}_{g,n}(X,\beta)\rightarrow \bar{M}_{g,n_2}(Y,g_*\beta),
\end{align}
so long as the space on the right exists.
(7)Under certain nice circumstances, if $M$ is convex, $\bar{M}_{0,n}(M,\beta)$ is non-singular of dimension
\begin{align}
\int_\beta c_1(T_M)+\operatorname{dim} M+n-3.
\end{align}

\begin{definition}
At each point $[\Sigma, p_1,\cdots,p_n,f]$ of $\bar{\mathcal{M}}_{g,n}(X,\beta)$, the cotangent line to $\sigma$ at point $p_i$ is a one dimensional vector space, which gives a line bundle $\mathbb{L}_i$, called the $i$th tautological line bundle.
\end{definition}
Traditionally, given classes $\gamma_1,\gamma_2,\cdots,\gamma_k\in H^*(M,\mathbb{Q})$,the gravitational descendant invariants are defined by
\begin{align}
& \langle \tau_{n_1}(\gamma_1)\tau_{n_2}(\gamma_2)\cdots\tau_{n_k}(\gamma_k)\rangle:\nonumber\\
 =& \sum_{A\in H_2(M,\mathbb{Z})}q^A \int_{[\bar{M_{g,k}}(M,A)]^{\text{Virt}}} c_1(\mathbb{L}_1)^{n_1}\cup \text{ev}^*_1(\gamma_1)\cup c_1(\mathbb{L}_2)^{n_2}\cup \text{ev}^*_2(\gamma_2)\cdots \nonumber\\
  & c_1(\mathbb{L}_k)^{n_k}\cup \text{ev}^*_1(\gamma_k)\rangle
\end{align}
The free energy $F_g^M$  can be written as
\begin{align}
F_g^M(t):=\langle \exp(\sum_{n,\alpha} t^\alpha_n \tau_n(\alpha))\rangle_g,
\end{align}
where $\mathcal{O}_1,\mathcal{O}_2,\cdots,\mathcal{O}_N$ form a basis of $H^*(M,\mathbb{Q})$; $\alpha$ ranges from $1$ to $N$; $n$ ranges over nonnegative integers; only finite $t^\alpha_n$ are nonzero.

In 1997,T. Eguchi, K. Hori and C. Xiong and S.Katz proposed a conjecture which generalized Witten 1990 conjecture \cite{E}:
Then the partition function is
\begin{align}
Z^M(t):=\exp(\sum_{g\geq 0}\lambda^{2-2g}F^M_g(t)). \label{partition}
\end{align}
The statement of Virasoro conjecture is that $Z^M(t)$ is  annihilated by $L_n$, $n\geq -1$, which forms part of Virasoro algebra with central charge $c=\chi(M)$,i.e., $\{L_n\}$ satisfy
\begin{align}
[L_n,L_m]=(n-m)L_{n+m}+\frac{\delta_{m,-n}}{12}\cdot \chi(M),
\end{align}
for $m,n\in\mathbb{Z}$. Since this conjecture was proposed, there have been lots of efforts on it. It has been confirmed up to genus 2 \cite{Lee} and there have been good results \cite{Liu1}\cite{FP}\cite{Liu2}\cite{Liu3} and etc.

   The representation of $L_n$, $n\geq -1$ is
\begin{align}
L_{-1}&=\sum_{\alpha=0}^N \sum_{m=1}^\infty m t^\alpha_m \partial_{m-1,\alpha}+\frac{1}{2\lambda^2}\sum_{\alpha=0}^N t^\alpha t_\alpha,\nonumber\\
L_0 &=\sum_{\alpha=0}^N \sum_{m=0}^\infty (m+b_\alpha)t^\alpha_m \partial_{m,\alpha}+(N+1)\sum_{\alpha=0}^{N-1}\sum_{m=0}^\infty m t^\alpha_m \partial_{m-1,\alpha+1}\nonumber\\
&+\frac{1}{2\lambda^2} \sum_{\alpha=0}^{N-1}(N-1)t^\alpha t_{\alpha+1}-\frac{1}{48}(N-1)(N+1)(N+3),\nonumber\\
L_n & =\sum_{m=0}^\infty \sum_{\alpha,\beta}\sum_j C^{(j)}_\alpha(m,n)(\mathcal{C}^j)_\alpha^\beta t^\alpha_m \partial_{m+n-j,\beta}\nonumber\\
& +\frac{\lambda^2}{2}\sum_{\alpha,\beta}\sum_j \sum_{m=0} D^{(j)}_\alpha(\mathcal{C})_\alpha^\beta \partial_m^\alpha \partial_{n-m-j-1,\beta}+\frac{1}{2\lambda^2}\sum_{\alpha,\beta}(\mathcal{C}^{n+1})_\alpha^\beta t^\alpha t_\beta,
\label{48}
\end{align}

\begin{align}
b_\alpha=q_\alpha-\frac{\operatorname{dim}M-1}{2},
\end{align}
where
\begin{align}
\mathcal{C}_\alpha^\beta=\int_M c_1(M)\wedge \omega_\alpha\wedge\omega^\beta,
\end{align}
and $\mathcal{C}^j$ is the $j$-th power of the matrix $\mathcal{C}$;
\begin{align}
C^{(j)}_\alpha(m,n)=& \frac{(b_\alpha+m)(b_\alpha+m+1)\cdots (b_\alpha+m+n)}{(m+1)(m+2)\cdots (m+n)}\nonumber\\
& \sum_{m\leq l_1<l_2<\cdots<l_j\leq m+n}\prod_j (\frac{1}{b_\alpha+l_j});
\end{align}
and
\begin{align}
D^{j}_\alpha(m,n)=&\frac{b^\alpha (b^\alpha+1)\cdots (b^\alpha+m)b_\alpha(b_\alpha+1)\cdots (b_\alpha+n-m-1)}{m! (n-m-1)!}\nonumber\\
& \sum_{-m\leq l_1<l_2<\cdots<l_j\leq n-m-1}\prod_j (\frac{1}{b_\alpha+l_j}
\end{align}

The operators  (\ref{48}) form a Virasoro algebra with a central charge $c=\sum_\alpha 1=\chi(M)$, if the following condition is satisfied
\begin{align}
\frac{1}{4}\sum_\alpha b^\alpha b_\alpha=\frac{1}{24}(\frac{3-\operatorname{M}}{2}\chi(M)-\int_M c_1(M)\wedge c_{\operatorname{M}M-1}(M)).
\end{align}
It was found that the above definition really forms a Virasoro algebra \cite{E}. It is instructive to verify it really forms an algebra here:

\begin{align}
& [L_{n_1},L_{n}]\nonumber\\
&=[\sum_{m_1=0}^\infty \sum_{\alpha_1,\beta_1}\sum_{j_1} C^{(j_1)}_{\alpha_1}(m_1,n_1)(\mathcal{C}^{j_1})_{\alpha_1}^{\beta_1} t^{\alpha_1}_{m_1} \partial_{m_1+n_1-j_1,\beta_1}\nonumber\\
& +\frac{\lambda^2}{2}\sum_{\alpha_1,\beta_1}\sum_{j_1} \sum_{m_1=0} D^{(j_1)}_{\alpha_1}(m_1,n_1)(\mathcal{C}^{j_1})_{\alpha_1}^{\beta_1} \partial_{m_1}^\alpha \partial_{n_1-m_1-j_1-1,\beta}\nonumber\\
&+\frac{1}{2\lambda^2}\sum_{\alpha_1,\beta_1}(\mathcal{C}^{n_1+1})_{\alpha_1}^{\beta_1} t^{\alpha_1} t_{\beta_1},
\sum_{m=0}^\infty \sum_{\alpha,\beta}\sum_j C^{(j)}_\alpha(m,n)(\mathcal{C}^j)_{\alpha}^{\beta} t^{\alpha}_{m} \partial_{m+n-j,\beta}\nonumber\\
& +\frac{\lambda^2}{2}\sum_{\alpha,\beta}\sum_j \sum_{m=0}^\infty D^{(j)}_\alpha(m,n)(\mathcal{C}^{j})_\alpha^\beta \partial_m^\alpha \partial_{n-m-j-1,\beta}\nonumber\\
&+\frac{1}{2\lambda^2}\sum_{\alpha,\beta}(\mathcal{C}^{n+1})_\alpha^\beta t^\alpha t_\beta]\nonumber\\
&=\sum_{m_1=0}^\infty \sum_{\alpha_1,\beta_1}\sum_{j_1} \sum_{m=0}^\infty \sum_{\alpha,\beta}\sum_j C^{(j_1)}_{\alpha_1}(m_1,n_1)\nonumber\\
& (\mathcal{C}^{j_1})_{\alpha_1}^{\beta_1}
C^{(j)}_\alpha(m,n)(\mathcal{C}^j)_{\alpha}^{\beta}(\delta_{\alpha,\beta_1}\delta_{m_1+n_1-j_1,m}
t^{\alpha_1}_{m_1}\partial_{m+n-j,\beta}-\delta_{m+n-j,m_1}\delta_{\beta,\alpha_1}t^\alpha_m \partial_{m_1+n_1-j_1,\beta_1} )\nonumber\\
&-\sum_{m_1=0}^\infty \sum_{\alpha_1,\beta_1}\sum_{j_1}\sum_{\alpha,\beta}\sum_j \sum_{m=0}^\infty\frac{\lambda^2}{2}C^{(j_1)}_{\alpha_1}(m_1,n_1)\nonumber\\
& (\mathcal{C}^{j_1})_{\alpha_1}^{\beta_1}D^{(j)}_\alpha(m,n)(\mathcal{C}^{j})_\alpha^\beta
(\eta^{\alpha\gamma}\delta_{m_1,m}\delta_{\gamma,\alpha_1}\partial_{n-m-j-1,\beta}\partial_{m_1+n_1-j_1,\beta_1}+\delta_{n-m-j-1,,m_1}\delta_{\alpha_1,\beta}\partial^\alpha_m \partial_{m_1+n_1-j_1,\beta_1})\nonumber\\
&+\sum_{m_1=0}^\infty \sum_{\alpha_1,\beta_1}\sum_{j_1}
\sum_{\alpha,\beta}\frac{1}{2\lambda^2}C^{(j_1)}_{\alpha_1}(m_1,n_1)(\mathcal{C}^{j_1})_{\alpha_1}^{\beta_1}(\mathcal{C}^{n+1})_\alpha^\beta (t^{\alpha_1}_{m_1}(\delta_{\beta_1,\alpha}t_{m_1+n_1-j_1,\beta}+\eta_{\beta,\gamma}t^{\alpha}_{m_1+n_1-j_1}\delta_{\beta_1,\gamma}))\nonumber
\end{align}
\begin{align}
&+ \sum_{m_1=0}^\infty \sum_{\alpha_1,\beta_1}\sum_{j_1}\sum_{\alpha,\beta}\sum_j \sum_{m=0}^\infty\frac{\lambda^2}{2}C^{(j)}_{\alpha}(m,n)\nonumber\\
& (\mathcal{C}^{j})_{\alpha}^{\beta}D^{(j_1)}_{\alpha_1}(m_1,n_1)(\mathcal{C}^{j_1})_{\alpha_1}^{\beta_1}
(\eta^{\alpha\gamma}\delta_{m,m_1}\delta_{\gamma,\alpha_1}\partial_{n_1-m_1-j_1-1,\beta}\partial_{m+n-j,\beta_1}+\delta_{n_1-m_1-j_1-1,,m}\delta_{\alpha_1,\beta}\partial^\alpha_{m_1} \partial_{m+n-j,\beta_1})\nonumber\\
&+\sum_{\alpha_1,\beta_1}\sum_{j_1} \sum_{m_1=0}
\sum_{\alpha,\beta}\frac{1}{4} D^{(j_1)}_{\alpha_1}(m_1,n_1)(\mathcal{C}^{j_1})_{\alpha_1}^{\beta_1}(\mathcal{C}^{n+1})_\alpha^\beta
(\delta_{\alpha,\beta}\delta_{m_1,n_1-m_1-j_1-1}+\eta^{\alpha\alpha}\eta_{\beta\beta}\delta_{m_1,n_1-m_1-j_1-1}\nonumber\\
&+\delta_{\alpha,\beta}t_{n_1-m_1-j_1-1}\partial_{m_1}^\alpha+t_{n_1-m_1-j_1-1}^\alpha \eta_{\beta\beta}\partial_{m_1}^\alpha+\eta^{\alpha\alpha}t_{m_1,\beta}\partial_{n_1-m_1-j_1-1,\beta}+\delta_{\alpha,\beta}t_{m_1}^\beta \partial_{n_1-m_1-j_1-1,\beta})\nonumber\\
&-\sum_{m=0}^\infty \sum_{\alpha_1,\beta_1}\sum_{j}
\sum_{\alpha,\beta}\frac{1}{2\lambda^2}C^{(j)}_{\alpha}(m,n)(\mathcal{C}^{j})_{\alpha}^{\beta}(\mathcal{C}^{n_1+1})_{\alpha_1}^{\beta_1} (t^{\alpha}_{m}(\delta_{\beta,\alpha}t_{m+n-j,\beta}+\eta_{\beta_1,\gamma}t^{\alpha_1}_{m+n-j}\delta_{\beta,\gamma}))\nonumber
\end{align}

\begin{align}
&-\sum_{\alpha_1,\beta_1}\sum_{j} \sum_{m=0}
\sum_{\alpha,\beta}\frac{1}{4}D^{(j)}_\alpha(m,n)(\mathcal{C}^j)_{\alpha}^{\beta}(\mathcal{C}^{n_1+1})_{\alpha_1}^{\beta_1}
(\delta_{\alpha_1,\beta_1}\delta_{m,n-m-j-1}+\eta^{\alpha_1\alpha_1}\eta_{\beta_1\beta_1}\delta_{m,n-m-j-1}\nonumber\\
&+\delta_{\alpha_1,\beta_1}t_{n-m-j-1}\partial_{m}^{\alpha_1}+t_{n-m-j-1}^{\alpha_1} \eta_{\beta_1\beta_1}\partial_{m}^{\alpha_1}+\eta^{\alpha_1\alpha_1}t_{m,\beta_1}\partial_{n-m-j-1,\beta_1}+\delta_{\alpha_1,\beta_1}t_{m}^{\beta_1} \partial_{n-m-j-1,\beta_1})\nonumber
\end{align}
\begin{align}
&=(n_1-n)(\sum_{m_1=0}^\infty \sum_{\alpha_1,\beta_1}\sum_{j_1} C^{(j_1)}_{\alpha_1}(m_1,n_1+n)(\mathcal{C}^{j_1})_{\alpha_1}^{\beta_1} t^{\alpha_1}_{m_1} \partial_{m_1+n_1+n-j_1,\beta_1}\nonumber\\
& +\frac{\lambda^2}{2}\sum_{\alpha_1,\beta_1}\sum_{j_1} \sum_{m_1=0} D^{(j_1)}_{\alpha_1}(m_1,n+n_1)(\mathcal{C}^{j_1})_{\alpha_1}^{\beta_1} \partial_{m_1}^\alpha \partial_{n_1+n-m_1-j_1-1,\beta}\nonumber\\
&+\frac{1}{2\lambda^2}\sum_{\alpha_1,\beta_1}(\mathcal{C}^{n_1+n+1})_{\alpha_1}^{\beta_1} t^{\alpha_1} t_{\beta_1})\nonumber\\
& =(n-n_1)L_{n+n_1}
\end{align}
Similar to \cite{E}, we have used identities:
\begin{align}
&\sum_{j_1=0}^j(\mathcal{C}^{j-j_1})_{\alpha_1}^{\beta_1}C^{(j-j_1)}_{\alpha_1}(m_1,n_1)C^{(j_1)}_{\beta_1}(m_1+n_1-j_1,n)\nonumber\\
&=(\mathcal{C}^{j-j_1})_{\alpha_1}^{\beta_1}((b_{\alpha_1}+m_1+n_1)C^{(j)}_{\alpha_1}(m_1,n_1+n)\nonumber\\
&+(m_1+n+n_1-j-j_1+1)C^{(j-1)}_{\alpha_1}(m_1,n_1+n));
\end{align}
and
\begin{align}
&\sum_{j_1=0}^j (\mathcal{C}^{j-j_1})_{\alpha}^{\alpha_1}D^{(j-j_1)}_\alpha(m,n)C^{(j_1)}_{\alpha_1}(n-m-j+j_1,n_1)\nonumber\\
&=(\mathcal{C}^{j-j_1})_{\alpha}^{\alpha_1}((b_\alpha+n-m-1)D^{(j)}_\alpha(k,n+n_1)+(n+n_1-m-j)D^{(j-1)}_\alpha(m,n+n_1)).
\end{align}

\subsection{the planar graph interpretation of Virasoro constraints}
The constraints $L_{-1}Z=0$ and $L_0 Z$ have been obtained in \cite{DW} and \cite{H}. We shall provide a planar graphic interpretation for $L_n Z=0$, $n\geq 1$. This interpretation is not rigorous so far.

    Two dimensional quantum gravity have been proved to be equivalent to hermitian matrix theories.
It is well known ,for a very general one matrix model is in fact a planar graph theory \cite{t'Hooft}\cite{B}\cite{Fran},
\begin{align}
Z_N(t_1,t_2,\cdots)&=\langle e^{N\sum_{i\geq 1}\operatorname{Tr}(M^i/i)}\rangle \nonumber\\
                   &=\sum_{n_1,n_2,\cdots\geq 0}\prod_{i\geq 1}\frac{(Nt_i)^{n_i}}{i^{n_i}n_i!}\langle \prod_{i\geq 1}\operatorname{Tr}(M^i)^{n_i}\rangle\nonumber\\
                   &=\sum_{n_1,n_2,\cdots\geq 0}\prod_{i\geq 1}\frac{(Nt_i)^{n_i}}{i^{n_i}n_i!} \nonumber\\
                   &\sum_{\text{ all labeled fat graphs $\Gamma$ with $n_i$ $i$-valent vertices }}
                    N^{-E(\Gamma)}N^{F(\Gamma)}\nonumber\\
                   &=\sum_{\text{fat graphs $\Gamma$}} \frac{N^{V(\Gamma)-E(\Gamma)+F(\Gamma)}}{\operatorname{Aut}(\Gamma)}\prod_{i\geq 1}g_i^{n_i(\Gamma)},
\end{align}
where $n_i(\Gamma)$ denotes the total number of $i$-valent vertices of $\Gamma$ and $V(\Gamma)=\sum_i n_i(\Gamma)$ is the total number of vertices of $\Gamma$.
We don't know if there is a matrix theory that is equivalent to the theory (\ref{partition}). However, we would like to propose

\textbf{Conjecture} \emph{The theory (\ref{partition}) is equivalent to a planar graph theory}.

This conjecture is of course weaker than the statement "the theory (\ref{partition}) is equivalent to a matrix theory." We expect that  this conjecture can be generalized to some other conformal theories.

In the rest of this subsection, we assume that this conjecture holds, then we shall show how the Virasoro constraints arise:
To show that the partition function (\ref{partition}) is annihilated by $L_n$, we will show that $L_n$ is the generator of particle antiparticle symmetry of the theory.

There are three terms in the right hand side of (\ref{48}), we claim that the first term
corresponds the annihilation of a vertex and creation of a vertex, such that the weight on the target
increases $n$. There are two factors to realize this: once a particle with target weight $m+n-j$ is annihilated and
a particle with target weight $m$ is created, from the graph, it is in fact a $m+n-j$ valent vertex   becomes a $m$ valent vertex. This makes $n-j$ to the contribution of the $n$ extra target weight. The rest $j$ target weight comes from directly from the factor $(\mathcal{C}^j)_\alpha^\beta$. The factor $D^{j}(m,n)$ is the ratio of the process from the graph:when $2(n-j)$ gluons or edges are decoupled from a $2m+2n-2j$ target weight vertex, to form the new graph, we need to divide it by $\frac{1}{(2m+2)(2m+4)\cdots (2m+2n-2j)}$. Mean while, when the new vertex is created, those gluons get new freedoms to couple(it is different from decoupling. When decoupling happens, there is no restriction) to the new vertex and the extra weights, we need a factor which is $(2q_\alpha+2m-(N-1))(2q_\alpha+2m+2-(N-1))\cdots (2q_\alpha+2m+2n-(N-1))(\sum_{m\leq l_1< l_2\cdots \l_j\leq m+n}\prod (\frac{1}{2q_\alpha+2l_j}-(N-1))$(there is $N-1$ but not $2(N-1)$ in these parenthesis because the maps in the moduli spaces are holomorphic). That is the explanation of the first term.

For the second term, the process is that two vertices form one more genus. This annihilation increases the target weight $(m+n-j-1+1)$. The last 1 in the parenthesis is because two vertices are annihilated.The rest of the factors have the similar explanation except the factor $\lambda^2$. In fact, this is because when the two vertices are annihilated,the Euler number increases by 2. The denominator 2 under $\lambda^2$ is because of the symmetry of the vertices.

The third term's explanation can be realized  similarly to the second term, by splitting a genus into two vertices. But the target weight increases totally from the $\mathcal{C}^j$.

Since $Z$ is invariant under conformal transformation, and according to our analysis above, $L_n$,  $n\geq 1$ is the representation of generator of this conformal transformation. Therefore $Z$ is annihilated by $L_n$, i.e., $L_n Z=0$.

\begin{remark}
this explanation of the Virasoro constraints is based on the existence of the corresponding planar graph theory. More knowledge about this conjectural planar graph theory is appreciated for deeper understanding of this explanation.
\end{remark}

\subsection{Virasoro constraints and conformal invariance}

Now let us turn to physics. To prove the partition function of a conformal field theory is annihilated by
the family $\{L_n\}$,  $n\geq -1$, one may wants to prove that function is invariant under the action of conformal transformation, which is a local coordinate transformation by a holomorphic function. Let us consider a sigma model from a Riemann surface to the smooth projective manifold $M$.
The Lagrangian  is
\begin{align}
\mathcal{L}&=\mathcal{L}_0-( (\int_\Sigma \sum_{m=0}^\infty \sum_\alpha t^\alpha_m\tau_m(\alpha)
 +\sum_{i=1}^r (\ln q_i) X^*(q_i))+(\ln \lambda )\chi(\Sigma)),
\end{align}
where $\sigma$ is the coordinate of the Riemann surface and  $g^{ab}$ is the metric on the Riemann surface;$\xi_m$s are fixed points on the Riemann surface;
 $X^\mu$ is the coordinate of the target space and $G_{\mu\nu}$ is the metric on the target space;
   $q_i$ in the front of $X(q_i)$ denotes  an indeterminate which corresponds a basis element of $H_2(M,\mathbb{Q})$ and  by abuse of notation, in $X(q_i)$, $q_i$ also represents the dual basis element in $H^2(M,\mathbb{Q})$;
 $\{\mathcal{O}_i\}$ form a basis of $H^*(M,\mathbb{Q})$.

Classically, this action is invariant under conformal transformation. If the genus $g$ is fixed, the path integral of the above classical action is the free energy for stable curve with genus $g$, so let us consider the partition function in which the genus $g$ is summed up:
we consider the following partition function
\begin{align}
Z(t)^M= \sum_{\Sigma\in S}\frac{1}{\operatorname{Aut}(\Sigma)}\int \frac{[dX dg]}{\text{Diff}\times\text{Weyl}}\exp(-\mathcal{L}). \label{partition2}
\end{align}
Here $S$ is the collection of finite stable curves.This partition function is exactly the partition defined in (\ref{partition}), by the definition of it.
The partition function in (\ref{partition}) is invariant under the conformal transformation. Therefore the partition function in (\ref{partition2}) is also invariant under conformal transformation.
conformal transformation can also be realized by local coordinate transformation
$$z'=z+\epsilon v(z)=z+\sum_{n=-1}^\infty \epsilon_n z^{n+1}
$$.
Since the conformal transformation is local, we would like to conjecture:

\emph{The representation of the generators of conformal transformation with central extension is} (\ref{48}).

 Locally a point in the moduli space can be represented as $(z,g_{ij}(z), X^\mu(z))$ (here we only write out the coordinate of one punched point) , which is an infinite dimensional space. We could define the hermitian metric on this infinite dimensional space, as the following
\begin{align}
\langle dz^a,dz^b\rangle&= g^{ab},\nonumber\\
\langle dz^c, dg^{ab}\rangle &=g^{cd}\partial_d g^{ab},\nonumber\\
\langle dg^{ab},dg^{cd}\rangle &=g^{ac}g^{bd},\nonumber\\
\langle  dz^a, dX^\mu\rangle &= g^{ab} G_{\mu\nu}\partial_b X^{\nu},\nonumber\\
\langle dX^\mu, dX^\nu \rangle &=G^{\mu\nu}.
\end{align}
where all the indexes run over complex coordinate index and their conjugates. Then it might be possible, but very complicated, to explicitly compute the first chern class of the vector bundle $\mathcal{L}_i$ and check the Virasoro condition.


\begin{thebibliography}{}
\bibitem{B}Bessis, D.; Itzykson, C.; Zuber, J. B.,Quantum field theory techniques in graphical enumeration.
Adv. in Appl. Math. 1 (1980), no. 2, 109--157. 
\bibitem{D} R. Dijkgraaf, H. Verlinde and E. Verlinde, Loop equations and Virasoro constraints
in nonperturbative two-dimensional quantum gravity, Nuclear Phys. B 348 (1991)
No. 3, 435-456.
\bibitem{DW}R. Dijkgraaf and E. Witten, Nucl. Phys. B342 (1990) 486.
\bibitem{Deligne}P.Deligne and D.Mumford, The irreduciblity of the space of curves of given genus, Inst.Hautes \'{E}tudes Sci. Publ.Math.45(1969)75.
\bibitem{E}Eguchi, Tohru; Hori, Kentaro; Xiong,Chuan-Sheng;,Quantum cohomology and Virasoro algebra,Phys. Lett. B 402 (1997), no. 1-2, 71--80.

\bibitem{F}M. Fukuma, H. Kawai and R. Nakayama, Continuum Schwinger-Dyson equations
and universal structures in two-dimensional quantum gravity, Internat. J. Modern
Phys. A6 (1991) No. 8, 1385-1406.

\bibitem{Fran}P.Di Francesco, 2D Quantum Gravity, Matrix Models and Graph Combinatorics, Applications of Random
Matrices in Physics.
\bibitem{Liu1}Liu, Xiaobo; Tian, Gang,Virasoro constraints for quantum cohomology.J. Differential Geom. 50 (1998), no. 3, 537--590.
\bibitem{FP}Hodge integrals and Gromov-Witten theory. Invent. Math. 139 (2000), no. 1, 173--199.
\bibitem{Liu2} Liu, Xiaobo, Elliptic Gromov-Witten invariants and Virasoro conjecture. Comm. Math. Phys. 216 (2001), no. 3, 705--728.
\bibitem{Liu3}Liu, Xiaobo, Genus-2 Gromov-Witten invariants for manifolds with semisimple quantum cohomology. Amer. J. Math. 129 (2007), no. 2, 463--498.
\bibitem{Lee}Lee, Y.-P.,Witten's conjecture and the Virasoro conjecture for genus up to two. Gromov-Witten theory of spin curves and orbifolds, 31--42, Contemp. Math., 403, Amer. Math. Soc., Providence, RI, 2006.
\bibitem{H}K. Hori, Nucl. Phys. B439 (1995) 395.
\bibitem{Harish}Harish-Chandra, Differential operators on a semisimple Lie algebra. Am. J. Math.
79, 87-120(1957).
\bibitem{Kac}V.G.Kac and A.K.Raina, Highest weight representation of infinite dimensional Lie algebras, Advanced Series in Mathematical PHysics Vol.2,World Scientific.
\bibitem{Kac1}V. Kac and A. Schwarz, geometric interpretation of the partition function of 2D
gravity, Phys. Lett. B 257 (1991) No. 3-4, 329-334.
\bibitem{K}Kontsevich,  intersection theory on the moduli space of curves and the Matrix Airy Function, Commun.Math.Phys. 147,1-23(1992).
\bibitem{I}Itzykson,Zuber,  Combinatorics of the modular group.II.the Kontsevich integrals, International Journal of modern physics A,vol 7,no.(23),5661-5705.
\bibitem{W1}E.Witten,Two dimensional gravity and intersection theory on moduli space, Surveys in Diff.Geom.1,243-310(1991).
\bibitem{P}Polchinski, String Theory Vol I,II, Cambridge University Press, 1998.
\bibitem{He}Singurdur Helgason, Groups and Geometric Analysis,Mathematical Surveys and Monographs Volume 83.
\bibitem{Mirror}Kentaro Hori, Sheldon Katz, Albrecht Klemm, and Rahul Pandharipande,Mirror Symmetry (Clay Mathematics Monographs, V. 1),American Mathematical Society; illustrated edition edition (July 2003).
\bibitem{M}L.M.Mehta, Random Matrices, third Edition.
\bibitem{t'Hooft}G.'t Hooft, Nuclear Phys. B72(1974)461.
 \bibitem{T}Vladimir Turaev,\emph{Introduction to Combinatorial Torsions}, Birkhauser (April 2001).
\bibitem{Tu}V.G.Turaev, \emph{Reidemeister torsion in knot theory} , Russian Math. Surveys 41,119-182,1986.
 \bibitem{J}Palle E.T.Jorgensen, \emph{Operators and Representation theory:Canonical Models for Algebras of Operators Arising Quantum Mechanics},
\bibitem{W1}Edward Witten,Two-Dimensional Gravity and Intersection Theory on Moduli Space, Surveys in Differential Geometry,1(1991)243-310.
\bibitem{W2}Edward Witten, Topological Sigma Model,Surveys in Diffrential Geometry 1(1991)243-310.
\bibitem{Z}D.P.\v{Z}elobenko, \emph{Compact Lie groups and their Representations}, Translations of Mathematical Monographs Volume 40.
\bibitem{GT}D.J.Gross and W.Taylor IV, Two dimensional QCD is a string theory, Nucl. Phys.B400(1993) 181-210,hep-th/9301068
\bibitem{Z}D.P.\v{Z}elobenko, \emph{Compact Lie groups and their Representations}, Translations of Mathematical Monographs Volume 40.
\end{thebibliography}
\end{document}